\documentclass[lettersize,journal]{IEEEtran}
\usepackage{amsmath,amsfonts,amsthm}
\usepackage{array}
\usepackage[caption=false,font=normalsize,labelfont=sf,textfont=sf]{subfig}
\usepackage{textcomp}
\usepackage{stfloats}
\usepackage{url}
\usepackage{verbatim}
\usepackage{graphicx}
\usepackage{cite}
\usepackage{wrapfig}
\usepackage{multirow}
\usepackage{booktabs}
\usepackage[ruled,linesnumbered,noend]{algorithm2e}

\newtheorem{definition}{Definition}
\newtheorem{theorem}{Theorem}
\newtheorem{lemma}{Lemma}

\begin{document}

\title{Adaptive DNN Surgery for Selfish Inference Acceleration with On-demand Edge Resource}

\author{Xiang~Yang,
        Dezhi~Chen,
        Qi~Qi,~\IEEEmembership{Senior~Member,~IEEE},
        Jingyu~Wang,~\IEEEmembership{Senior~Member,~IEEE},
        Haifeng~Sun,
        Jianxin~Liao,
        and~Song~Guo,~\IEEEmembership{Fellow,~IEEE}%
        \IEEEcompsocitemizethanks{
                \IEEEcompsocthanksitem Xiang Yang, Dezhi Chen, Qi Qi, Jingyu Wang, Haifeng Sun and Jianxin Liao are with the State Key Laboratory of Networking and Switching Technology, Beijing University of Posts and Telecommunications.
                E-mail: \{yangxiang, chendezhi, qiqi8266, wangjingyu, hfsun, liaojx\}@bupt.edu.cn.
                \IEEEcompsocthanksitem Song Guo is an IEEE Fellow (Computer Society) and an ACM Distinguished Member with the Department of Computing at The Hong Kong Polytechnic University.
                E-mail: cssongguo@comp.polyu.edu.hk
                \IEEEcompsocthanksitem Qi Qi and Jingyu Wang are the corresponding authors.}
}

\maketitle

\begin{abstract}
        Deep Neural Networks (DNNs) have significantly improved the accuracy of intelligent applications on mobile devices. DNN surgery, which partitions DNN processing between mobile devices and multi-access edge computing (MEC) servers, can enable real-time inference despite the computational limitations of mobile devices. However, DNN surgery faces a critical challenge: determining the optimal computing resource demand from the server and the corresponding partition strategy, while considering both inference latency and MEC server usage costs. This problem is compounded by two factors: (1) the finite computing capacity of the MEC server, which is shared among multiple devices, leading to inter-dependent demands, and (2) the shift in modern DNN architecture from chains to directed acyclic graphs (DAGs), which complicates potential solutions.

        In this paper, we introduce a novel Decentralized DNN Surgery (DDS) framework. We formulate the partition strategy as a min-cut and propose a resource allocation game to adaptively schedule the demands of mobile devices in an MEC environment. We prove the existence of a Nash Equilibrium (NE), and develop an iterative algorithm to efficiently reach the NE for each device. Our extensive experiments demonstrate that DDS can effectively handle varying MEC scenarios, achieving up to 1.25$\times$ acceleration compared to the state-of-the-art algorithm.
\end{abstract}

\begin{IEEEkeywords}
        Article submission, IEEE, IEEEtran, journal, \LaTeX, paper, template, typesetting.
\end{IEEEkeywords}

\section{Introduction}
\IEEEPARstart{O}{ver} the past several years, a multitude of intelligent applications have emerged on mobile devices, including unmanned aerial vehicles and wearable devices. These developments are attributed to advances in deep neural network (DNN) models \cite{iot:uav,koch2019reinforcement,self-driving,ubi-ar}. However, due to the computational limitations of mobile devices, on-device DNN inference is often challenging. A common workaround, referred to as \textit{DNN surgery}, involves distributing the inference task between the mobile device and a cloud server equipped with high-performance GPUs \cite{sigarch:neuro,infocom19:mincut,ubicom20:faster-infer,9155237}. Under this scheme, mobile devices gather and process raw data using early DNN layers to produce intermediate results (\emph{features}). These features are then forwarded to the cloud, which completes the final computations and returns the inference results to the mobile device.

DNN surgery partitions the DNN model for deployment on both the edge and the cloud, thereby capitalizing on the server's computational resources. In comparison to on-device inference \cite{chen2016eyeriss,yang2017designing}, DNN surgery accelerates the inference task on mobile devices. Furthermore, when compared to directly uploading raw data to the cloud \cite{ubi-ar}, DNN surgery compresses and encrypts the raw data via the processing of initial neural layers, thereby enhancing user privacy.

However, DNN surgery does introduce communication costs between the mobile device and the cloud, which can impact overall inference latency. Network latency can be substantial between mobile devices and the cloud, and bandwidth is often limited \cite{zhou2019edge}, leading to uncertainties in real-time inference. Conversely, Multi-access Edge Computing (MEC) offers cloud computing capabilities at the edge of the cellular network, thereby reducing latency and providing high bandwidth. By moving services and applications closer to the user, MEC improves user experiences, particularly for delay-sensitive applications \cite{kekki2018mec,8736011,ft:survey}.

In this study, we enhance the DNN surgery technology by adapting it to the MEC scenario to ensure real-time inference. Unlike an exclusively cloud-based infrastructure, the key distinction is the highly \emph{on-demand} nature of MEC server service. Given that the MEC server only caters to the computing needs of nearby mobile devices, and the locations of these devices are uncertain, it cannot continuously reserve idle resources for any specific device. Consequently, the computing resources that the current MEC server can provide for an inference request are unpredictable. However, prior research \cite{sigarch:neuro,infocom19:mincut,ubicom20:faster-infer,8941306,ubi-edge-cloud} has uniformly assumed a fixed amount of server-provided computing resources for each mobile device, which is challenging to attain in the MEC scenario.

\begin{figure*}
        \centering
        \includegraphics[width=0.8\linewidth]{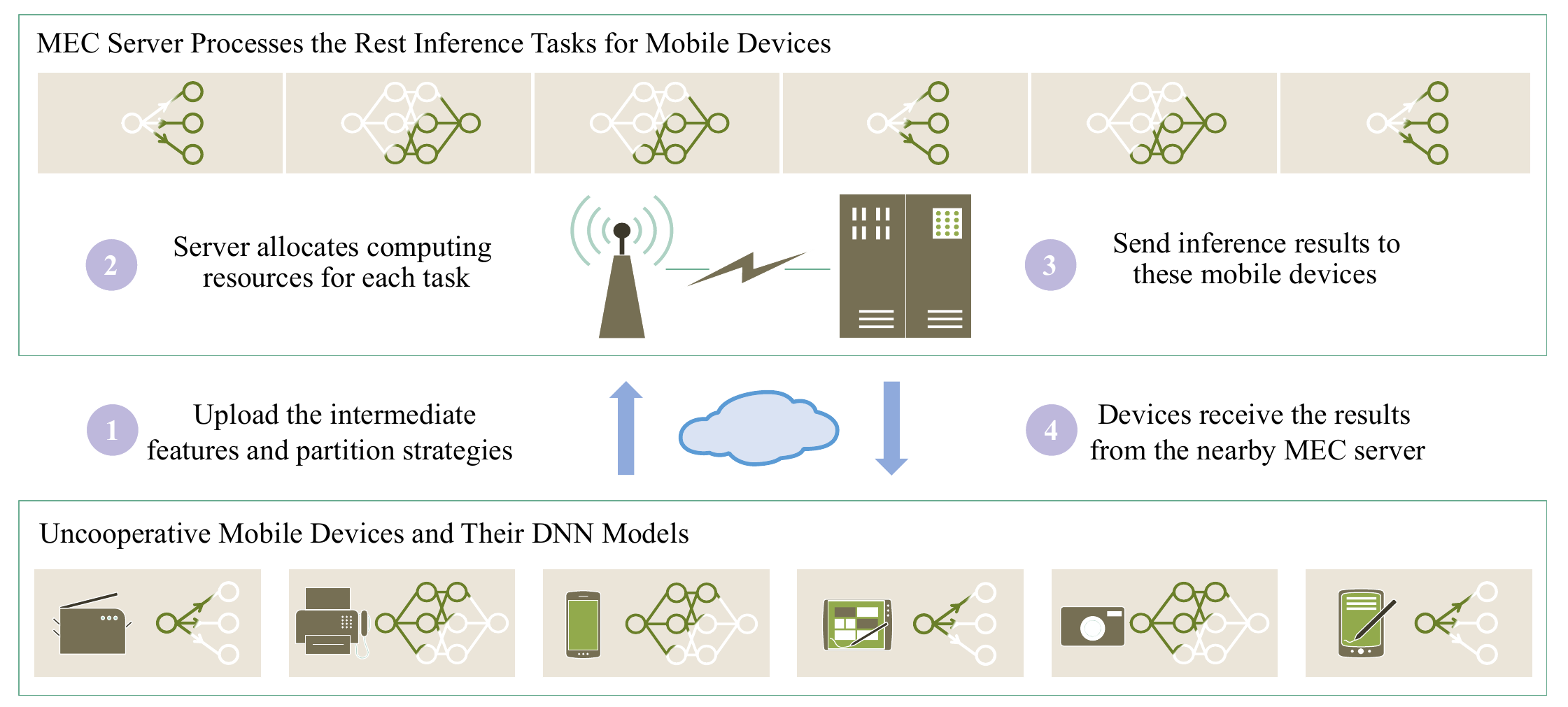}
        \caption{Illustration of a typical MEC scenario where multiple independent mobile devices collectively utilize the computing resources provided by a single MEC server.}
        \label{fg:network-topology}
\end{figure*}

Fig. \ref{fg:network-topology} depicts a typical MEC scenario where multiple independent mobile devices are connected to the same MEC server, collectively utilizing its computing resources. Each device executes a distinct inference task. For instance, a virtual reality device may run YOLO \cite{yolo} for object detection, while a smartphone might employ ResNet \cite{resnet} for junk mail classification. These devices pay charges to leverage the MEC server's computing resources and accelerate their respective inference tasks based on specific requirements. Each device aims to minimize its overall inference cost, including the inference latency, associated cost of using the MEC server and so on. During the inference process, while there is no direct interaction between devices, their fluctuating demands (required computing resources) impact the performance of other devices sharing the MEC server's resources.

Implementing DNN surgery within the MEC scenario presents several challenges. The mobile environment is characterized by a high level of complexity, with devices utilizing various DNN models, a diversity of computing resources and uplink rates, and ultimately, distinct demands. Moreover, each mobile device only care about its overall inference cost, and can enter or exit the MEC server's service area unpredictably. (1) \emph{The centralized scheduling of these selfish mobile devices to accommodate their varying demands proves challenging.} Preceding algorithms \cite{infocom19:mincut,ubicom20:faster-infer} necessitate each mobile device to have precise knowledge of the available computing resources that the server can provide. However, the server's workload fluctuates constantly due to the dynamic demands of mobile devices. (2) \emph{For mobile devices, making decisions under unpredictable server workload conditions is difficult.} Further, the varying required computing resources impact the choice of DNN partition strategy, and vice versa. Additionally, the characteristics of DNN models used in inference tasks pose further obstacles. Firstly, DNNs comprise a broad variety of neural layers, each with distinct computational and communication overheads. Secondly, contemporary DNN models are structured as directed acyclic graphs (DAGs), not chains. Therefore, (3) \emph{determining the optimal required computing resources and partition strategy for a given mobile device at the same time is complex.}

We propose the \emph{Decentralized DNN Surgery} (DDS) framework to address these challenges. DDS ascertains each mobile device's demand through a resource allocation game. In this game, each device requests a portion of computing resources from the MEC server at a base price. If the total requested resources surpass the MEC server's capacity, the server raises its service price for profit purposes. Consequently, each mobile device's demand influences others through the MEC server's service price, necessitating a compromise among devices. Each device's DNN and corresponding partition strategy are formulated as a DAG and an edge cut of the DAG, respectively. The allocated computing resources from the MEC server are employed to identify the best partition strategy by establishing a latency graph based on the original DNN graph and applying a min-cut algorithm on the latency graph. This paper demonstrates the equivalence between the minimum edge cut and the optimal partition strategy. We deduce the existence of a unique Nash Equilibrium (NE) for all mobile devices and propose an iterative algorithm executed on every mobile device to efficiently converge to the NE.

DDS is a lightweight framework necessitating minimal communication among mobile devices and virtually no computational overhead. Furthermore, DDS has an $O(1)$ computational complexity as the number of mobile devices participating in the resource allocation game increases. Thus, it can be readily deployed in most real-world applications. We evaluate DDS on a simulated cluster with up to 100 mobile devices. The experimental results reveal that DDS improves throughput by $1.25 \times$ compared to state-of-the-art algorithms.

In summary, the following contributions are offered in this paper:
\begin{itemize}
        \item Adaptation of DNN surgery to MEC scenarios where multiple mobile devices share computing resources, filling a gap in the current literature on DNN surgery algorithms.
        \item Formulation of DNN surgery with multiple devices as a resource allocation game, from which the existence of a unique NE is deduced.
        \item Development of the Decentralized DNN Surgery (DDS) framework which adaptively selects the best demand for computing resources and the optimal partition strategy for each mobile device.
        \item Evaluation of the proposed DDS framework through extensive simulations, demonstrating superior performance in comparison to existing algorithms.
\end{itemize}

\section{Background}
In this section, we provide a brief overview of the essential concepts related to the DNN model and DNN surgery.

\subsection{DNN as a Graph}

Early DNNs, such as LeNet \cite{lenet}, AlexNet \cite{alexnet}, and VGG \cite{vgg}, were composed of a series of contiguous basic neural layers, including convolutional layers, pooling layers, and fully connected layers. However, to achieve superior performance, modern DNN architectures, including ResNet \cite{resnet}, Inception \cite{inception}, and Visual Transformer \cite{vit}, are constructed from a sequence of meticulously designed blocks, superseding the use of individual neural layers. As a result, these model structures bear a greater resemblance to directed acyclic graphs (DAGs) than to chains, as illustrated in Fig. \ref{fg:graph-dnn}. The left part of the figure depicts the chain structure of VGG16, while the right side showcases the graph structures of Inception \cite{inception}. Each neural layer in a chain structure has, at most, one input and output layer. In contrast, layers in a graph structure can connect to multiple input or output layers. As such, the number of potential partition strategies for DNNs with graph structures significantly exceeds that for DNNs with simple chain structures.
\begin{figure}
        \centering
        \includegraphics[width=0.9\linewidth]{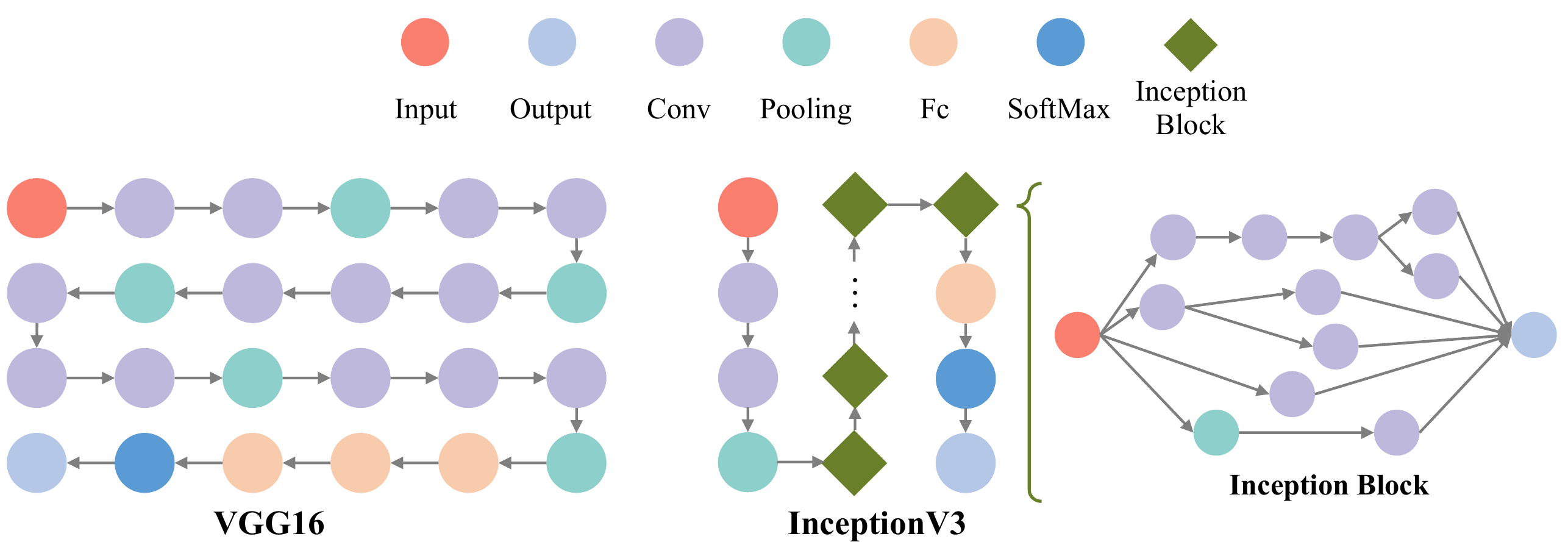}
        \caption{DNNs with chain and graph structures.} \label{fg:graph-dnn}
\end{figure}

\subsection{DNN Surgery}
Mobile devices can collect vast amounts of data but often lack sufficient computational resources to execute real-time inference locally. Conversely, uploading raw data to server introduces impractical network latency. DNN surgery addresses this issue by offloading part of the DNN inference computation to the mobile device. As raw data passes through neural layers, features are extracted and become significantly smaller in size than the original data \cite{infocom19:mincut,icdcs:pico}. Therefore, uploading intermediate features incurs lower network latency compared to raw data. Moreover, this approach reduces the computational load on the MEC server and encrypts the original data within the intermediate features, thereby preserving user privacy. However, modern DNN structures are highly complex, and computation and communication overheads vary across different neural layers. Given the diverse requirements of numerous mobile devices, the main challenge lies in determining the optimal DNN partition strategy.

\section{System Model} \label{sec:cost-model}

This section outlines the system model in three distinct steps: (1) Formulating the cost function for the inference task when an optimal partition strategy is given; (2) Representing the partition strategy as a cut and identifying the optimal partition with a fixed allocation of computing resources from the MEC server; (3) Proposing a resource allocation game to determine the optimal resource demand for each mobile device.

\subsection{Cost Function for Inference}

We utilize a directed acyclic graph (DAG), denoted as $\mathcal{G} : <\mathcal{V}, \mathcal{E}>$, to construct the DNN employed by the mobile device. Here, $\mathcal{V} : \left\{v_1, \cdots, v_i, \cdots, v_n \right\}$ represents the neural layers within $\mathcal{G}$, with $v_1$ and $v_n$ symbolizing the source and sink vertices, respectively. Furthermore, $\mathcal{E} : \left\{ e:(u, v) | u, v \in \mathcal{V} \right\}$ signifies the data flow within $\mathcal{G}$, where an edge $(u, v)$ indicates that the output of layer $u$ serves as the input for layer $v$.

DNN surgery involves formulating a partition strategy to identify an edge cut, denoted as $\mathcal{C}$. Removing $\mathcal{C}$ splits the original graph $\mathcal{G}$ into two disjoint segments. The segment containing the source vertex is deployed on local mobile devices ($\mathcal{G}^l$), while the segment containing the sink vertex is executed on the MEC server ($\mathcal{G}^s$). Note that $\mathcal{G}^l$ and $\mathcal{G}^s$ may be empty, indicating that the entire DNN is executed either on the MEC server or the local mobile device.

Upon receiving raw input data, the mobile device initially processes the data locally through its partial DNN ($\mathcal{G}^l$), then forwards the intermediate feature and partition strategy ($\mathcal{C}$) to the server. The MEC server computes the final output of $\mathcal{G}$ using $\mathcal{G}^s$ and dispatches the result back to the mobile device. The total inference latency ($T$) consists of three components: the local mobile device's inference latency ($T^l$), the transmission time of the intermediate feature ($T^t$), and the server-side inference latency ($T^s$). Analogous to \cite{infocom19:mincut,icdcs:pico}, we train a regression model for each device, considering both the required computing resources of different layers and the computing capacity of the device to predict the execution time.

These components are formulated as follows. Let $\mathcal{D} : \left\{ d_1, d_2, \cdots, d_n \right\}$ represent all mobile devices utilizing the same MEC server ($s$). The term $E(d_i)$, where $E: \mathcal{D} \to R$, denotes the computing capacity of device $d_i$, expressed in \emph{FLOPS} (Floating Operations Per Second). Additionally, $f(v_i)$, where $f: \mathcal{V} \to R$, signifies the required \emph{FLOPs} (Floating Operations) of a neural layer $v_i$. The edge computing latency for device $d_i$ can be represented as:
\begin{equation}
        T^l(d_i) := \alpha_i \frac{\sum_{v_j \in \mathcal{G}^l} f(v_j)}{E(d_i)}
\end{equation}
where $\alpha_i$ is the parameter of the trained regression model.

In a similar vein, let $g(d_i)$, where $g: \mathcal{D} \to {R}$, denote that device $d_i$ has allocated $g(d_i)$ computing resources from the MEC server. The computing latency $T^s$ on the server is:
\begin{equation} \label{eq:tc}
        T^s(d_i, g(d_i)) := \alpha_s \frac{\sum_{v_j \in \mathcal{G}^s} f(v_j)}{g(d_i)}.
\end{equation}
Here, $\alpha_s$ represents the regression model parameter on the server.

Regarding the transmission time $T^t$, it is expressed as follows:
\begin{equation}
        T^t(d_i) := \frac{\sum_{e_i \in \mathcal{C}_i} o(e_i)}{b(d_i)}
\end{equation}
where $o: \mathcal{E} \to R$ denotes the feature size of each edge $e_i$ and $b: \mathcal{D} \to R$ is the bandwidth of device $d_i$.

The total inference latency $T$ under a specific partition strategy $\mathcal{C}$ can be represented as $T := T^l + T^t + T^s$. If $g(d_i)$ is fixed, there exists an optimal partition strategy $\mathcal{C}^\star$ which minimizes the total inference latency $T$. We denote the minimum inference latency as $T^\star$.

Due to the service of the MEC server, each mobile device also incurs a charge for the service. Let $\beta(g(d_i))$ represent the charge if $g(d_i)$ computing resources are used. The overall cost function of the inference task for each mobile device can be formulated as:
\begin{equation} \label{eq:cost-fucntion-naive}
        \mathcal{L}(d_i, g(d_i)) := T^\star(d_i, g(d_i)) + \gamma \beta(g(d_i)),
\end{equation}
which comprises the optimal inference latency and the service charge. Here, $\gamma$ is a hyperparameter.

\subsection{Defining the Partition Strategy as a Cut} \label{sec:graph-mincut}

\begin{figure}
        \centering
        \includegraphics[width=0.9\linewidth]{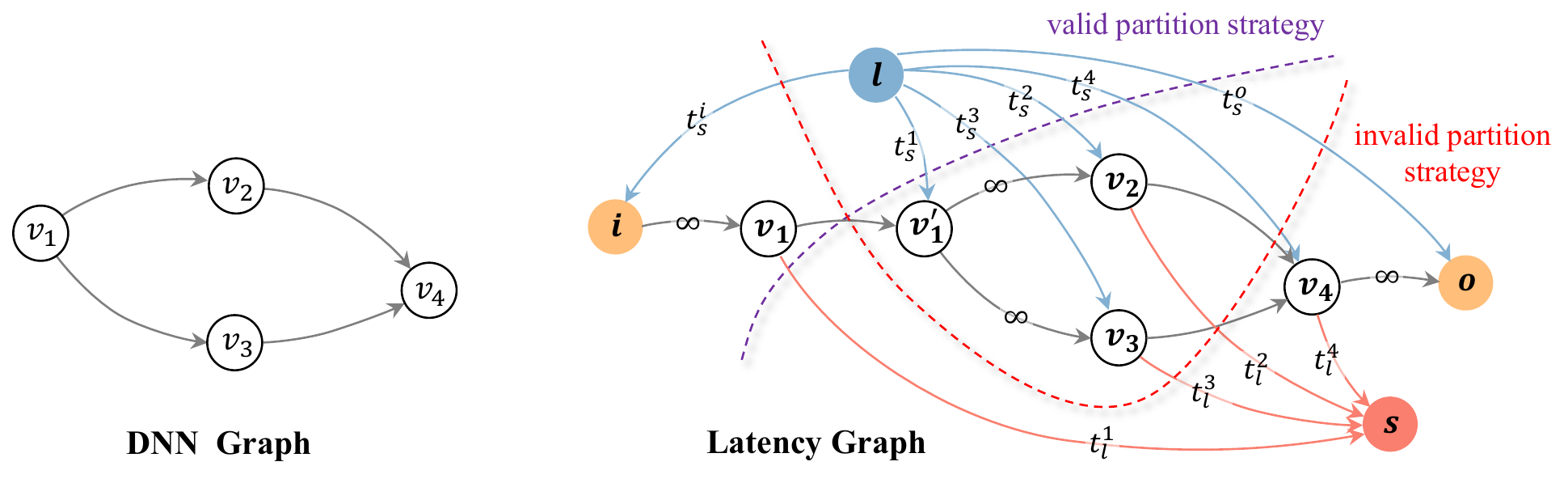}
        \caption{Illustration of latency graph $\mathcal{\hat{G}}$. $t_s$ represents the server inference latency for layer $v_i$, while $t_l$ denotes the local inference latency on the mobile device. $t_s^i$ and $t_s^o$ denote the transmission time of input data and inference result, respectively.}
        \label{fg:mincut-describe}
\end{figure}

Identifying the optimal inference latency $T^\star$ is equivalent to determining the optimal partition strategy $\mathcal{C}^\star$. This can be achieved by constructing a specific directed acyclic graph (DAG) $\mathcal{\hat{G}} : <\mathcal{\hat{V}}, \mathcal{\hat{E}}, \phi>$ derived from the original DAG $\mathcal{G}$, where $\phi: \mathcal{\hat{E}} \to R$ represents the capacity of every edge. We refer to $\mathcal{\hat{G}}$ as the \emph{latency graph} and construct it using the following steps:
\begin{enumerate}
        \item For each edge $e_i \in \mathcal{E}$, assign its capacity as $\frac{o(e_i)}{b(d_i)}$, which corresponds to the network latency for uploading the feature of edge $e_i$.
        \item For every vertex $v$ with multiple outgoing vertices, split $v$ into two vertices $v'$ and $v$. Add a new edge $(v', v)$ and set its capacity to $o((v, u))$. For each $(v, u)$, adjust its capacity to $\infty$.
        \item Introduce four virtual vertices $l$, $s$, $i$, and $o$. The vertices $l$ and $s$ denote the local mobile device and the MEC server, respectively. The vertices $i$ and $o$ are special vertices representing the captured raw data and the final output, respectively.
        \item Add two edges $(l, i)$ and $(i, v_1)$. The capacity of $(l, i)$ is the network latency for uploading the entire raw data, and the capacity of $(i, v_1)$ is set to $\infty$. Add two edges $(l, o)$ and $(v_n, o)$. The capacity of $(l, o)$ is the network latency for sending the final inference result to the mobile device, and the capacity of $(o, v_n)$ is set to $\infty$.
        \item For each vertex $v_i \in \mathcal{G}$, add two edges $(l, v_i)$ and $(v_i, s)$. The capacity of edge $(l, v_i)$ is assigned to $\alpha_s \frac{f(v_i)}{g(d_i)}$, representing the computing latency for layer $v_i$ on server $s$. For edge $(v_i, s)$, its capacity is assigned to $\alpha_i \frac{f(v_i)}{E(d_i)}$, denoting the local computing latency.
\end{enumerate}

It is worth noting that compared with \cite{infocom19:mincut}, we add new vertices $i$, $o$ and related edges. The vertex $i$ and edge $(i, v_1)$, whose capacity is $\infty$, ensure that the raw data is always captured in the mobile device. Similarly, the edges $(l, o)$ and $(v_n, o)$ ensure that the MEC server returns the inference results once the final layer $v_n$ is deployed on the MEC server. Furthermore, any vertex $v$ with multiple successors is replaced with $(v, v')$ to prevent redundant calculations during the estimation of network transmission time for $v$. Without these constraints, the partition strategy could be suboptimal.

The construction of a latency graph can be elucidated with a simple DNN graph with four vertices, as depicted in Fig. \ref{fg:mincut-describe}. Vertices $v_1$ and $v_4$ act as the source and sink, respectively. The blue edge cut in Fig. \ref{fg:mincut-describe} represents a partition strategy where $v_1$ are executed on the mobile device, while $v_2$, $v_3$ and $v_4$ are processed on the server. The capacity of $(v_1, v_1')$ equates to the network latency of the partition strategy. Notably, we replace $v_1$ with $(v_1, v_1')$, differing from the original DNN graph, to prevent the blue edge cut from passing through $(v_1, v_2)$ and $(v_1, v_3)$. The cumulative capacity of these edges would be twice the actual transmission. The edge cut includes edges $(v_1, s)$, whose capacity represents the local inference latency $T^l$. The remaining edges, $(l, v_2)$, $(l, v_3)$ and $(l, v_4)$, denote the server latency. Hence, the inference latency $T$ can be expressed as a summation of edge capacities:
\begin{equation}
        T(d_i, g(d_i)) = \sum_{e_i \in \mathcal{\hat{C}}} \phi(e_i),
\end{equation}
where $\mathcal{\hat{C}}$ is an edge cut of latency graph $\mathcal{\hat{G}}$. The corresponding partition strategy $\mathcal{C}$ can be symbolized as $\mathcal{\hat{C}} \cap ( \mathcal{E} \cup { (i, v_1) } )$.

Each partition strategy $\mathcal{C}$ for model $\mathcal{G}$ corresponds to an edge cut $\mathcal{\hat{C}}$ in the latency graph $\mathcal{\hat{G}}$. The corresponding inference latency of $\mathcal{C}$ is the sum of the capacities of edges in $\mathcal{\hat{C}}$. However, it cannot be assured that any edge cut of latency graph $\mathcal{\hat{G}}$ is a valid partition for the original graph $\mathcal{G}$. The edge cut represented by the red line in Fig. \ref{fg:mincut-describe} also segregates $\mathcal{\hat{G}}$ into two disjoint parts that contain either the source or sink vertex, thus constituting a valid edge cut. However, for the original graph $\mathcal{G}$, it places $v_1$ on the server but its successor $v_2$ on the mobile devices. Such a partition strategy lacks practical feasibility.

To address this, we introduce the concept of a valid cut:
\begin{definition}[Valid Cut] \label{def:valid-cut}
        Let $\mathcal{\hat{C}}$ be an edge cut for the latency graph $\mathcal{\hat{G}}$ that partitions the original DNN $\mathcal{G}$ into two segments, $\mathcal{G}^l$ and $\mathcal{G}^s$. We deem $\mathcal{\hat{C}}$ a valid cut if $u \in \mathcal{G}^s \Rightarrow v \in \mathcal{G}^s, \forall (u, v) \in \mathcal{\hat{C}} \cap \mathcal{E}$.
\end{definition}

Given that the objective of DNN surgery is to minimize inference latency, the minimum edge cut $\mathcal{\hat{C}^\star}$ of graph $\mathcal{\hat{G}}$ is the primary concern. Fortunately, we can establish that:
\begin{theorem} \label{th:divide}
        If $\mathcal{\hat{C}^\star}$ is the minimum edge cut of the latency graph $\mathcal{\hat{G}}$, then it must also be a valid cut for the DNN graph $\mathcal{G}$.
\end{theorem}
\begin{proof}
        A simple proof sketch is provided here. An important property of $\mathcal{\hat{C}^\star}$ is that for all vertices $v_i \in \mathcal{G}$, the ratio $\frac{\phi((l, v_i)) }{\phi((v_i, s))}$ equals $ \frac{\alpha_i E(d_i)}{\alpha_s g(d_i)}$. We examine this property in two scenarios.

        If $\alpha_i E(d_i) \geq \alpha_s g(d_i)$, the minimum edge cut is equivalent to the set comprising all edges connected to $s$. This conforms to our definition and implies that the entire graph is placed on the local device.

        If $\alpha_i E(d_i) \geq \alpha_s g(d_i)$, suppose $\mathcal{\hat{C}^\star}$ is the minimum edge cut. If $\mathcal{\hat{C}^\star}$ is not a valid cut, there must exist an edge $(u, v) \in \mathcal{\hat{C}^\star}$ such that $u \in \mathcal{G}^s$ and $v \in \mathcal{G}^l$ according to Definition \ref{def:valid-cut}. As $u, v \in \mathcal{V}$ and $\mathcal{G}$ is a DAG, there must exist a path $p = {v, \cdots, v_n}$ that starts from $v$, with $v_n$ being the sink vertex of $\mathcal{G}$. This path $p$ can be extended into a smaller DAG with $v$ as the source vertex, which we denote as $\mathcal{G}' = <\mathcal{V}', \mathcal{E}'>$. Since $v \in \mathcal{G}^l$, we let $\Delta \mathcal{V} = \mathcal{V}' \cap \mathcal{V}^l$ denote the vertices that are part of both $\mathcal{G}'$ and $\mathcal{G}^l$. We then construct a new edge cut, $\mathcal{\hat{C}}'$, by removing $\Delta \mathcal{V}$ from $\mathcal{G}^l$ and adding it to $\mathcal{G}^s$. The new edge cut $\mathcal{\hat{C}}'$ adds edges $(l, v_i)$ and removes edges $(v_i, v_n)$ for all $v_i \in \Delta \mathcal{V} $. It also removes edges that connect $\Delta \mathcal{V}$ to the remaining part of the original graph $\mathcal{G}$, representing the additional network latency. As $\alpha_i E(d_i) \geq \alpha_s g(d_i)$, the total capacity of the removed edges exceeds that of the added edges. Consequently, the total capacity of edges in $\mathcal{\hat{C}}'$ is smaller than $\mathcal{\hat{C}^\star}$. However, this conclusion contradicts our assumption that $\mathcal{\hat{C}^\star}$ is the minimum edge cut. Therefore, $\mathcal{\hat{C}^\star}$ must be a valid cut.
\end{proof}

In the light of Theorem \ref{th:divide}, we can assure that executing a min-cut algorithm on the latency graph $\mathcal{\hat{G}}$ derived from a given DNN $\mathcal{G}$ will yield a valid and optimal partition strategy that minimizes the inference latency $T$ if $g(d_i)$ is provided.

\subsection{ DNN Surgery in the Presence of Competing Devices}\label{sec:problem-definition}

In a practical MEC scenario, multiple non-cooperative mobile devices are linked to the same MEC server, which possesses finite computing resources, denoted by $S$. Each mobile device $d_i$ has an optimal demand $g^{\star}(d_i)$ of computing resources that minimizes its individual cost function $\mathcal{L}(d_i, \cdot)$. However, given the finite computing capacity $S$ of the MEC, the demands $g^{\star}i(d_i)$ may not be satisfiable for all devices when $S < \sum_{d_i \in \mathcal{D}} g^{\star}_i(d_i)$. Economically, the MEC server is inclined to increase the service price when demand surpasses supply, primarily for profit maximization (for instance, the on-demand instance provided by AWS).

Suppose the base unit price of computing resources provided by the MEC server is $1$, and each mobile device $d_i$ has a budget $a_i$ for renting computing resources. If the sum of $a_i$ from all devices surpasses the limit $S$, the current MEC server's corresponding business entity will proportionally elevate the price to $\frac{\sum_{a_i \in \mathcal{A}}a_i}{S}$. Here, $\mathcal{A} = \left\{a_1, \cdots, a_N\right\}$ constitutes the set of all the budgets from mobile devices. Consequently, the actual allocated computing resource $g(a_i, \mathcal{A})$ for device $d_i$ from MEC server can be expressed as follows:
\begin{equation} \label{eq:res-proportion}
        {g}(a_i, \mathcal{A}) := \left\{\begin{array}{cl}
                \frac{ a_i }{ \sum_{a_j \in \mathcal{A}} a_j }S & \sum_{a_j \in \mathcal{A}} a_j > S \\
                a_i                                             & otherwise .
        \end{array}\right.
\end{equation}
To simplify the following description, we denote $g(a_i, \mathcal{A})$ as $g_i$. The cost function Eq. \eqref{eq:cost-fucntion-naive} under the dynamic pricing scenario can be rewritten as:
\begin{equation} \label{eq:cost-function-game}
        \mathcal{L}(d_i, a_i, \mathcal{A}) := T^\star(d_i, g_i) + \gamma a_i .
\end{equation}
The charge function $\beta(\cdot)$ in Eq. \eqref{eq:cost-fucntion-naive} as a linear function in this context.

Every mobile device is self-interested, aiming to minimize their cost function $\mathcal{L}(d_i, a_i, \mathcal{A})$ without concerning the performance of others. However, each self-centered mobile device has to find a compromise with others when $S < \sum_{d_i \in \mathcal{D}} g^{\star}_i(d_i)$, since the selection of $a_i$ influences each other through $\mathcal{A}$.

For instance, the unit price $\frac{\sum_{a_i \in \mathcal{A}}a_i}{S}$ will escalate if all devices opt to raise their budgets to acquire more computing resources, consequently augmenting the cost for every mobile device. In such a case, the problem evolves into a non-cooperative game involving $N$ players.
\begin{definition}[Resource Allocation Game $\Gamma$] \label{def:game}
        $N$ players allocate resources from a limited source $S$. All players choose their budgets $a_i$, and the final allocated resources $g_i$ for them follow Eq. \eqref{eq:res-proportion}. Each player's objective is to minimize its cost function $\mathcal{L}(d_i, a_i, \mathcal{A})$.
\end{definition}

Different mobile devices have varying $a_i$. Devices with lightweight DNNs tend to utilize fewer computing resources to diminish service charges, while others may allocate more resources to expedite inference tasks. A crucial problem in game $\Gamma$ is the existence of a Nash Equilibrium that accommodates the demands of heterogeneous mobile devices:
\begin{definition}[Nash Equilibrium]
        A strategy vector $\mathcal{A}^\star = \left\{ a^{\star}_{1} , \cdots , a^{\star}_{N} \right\}$ is considered a Nash Equilibrium (NE) of the defined game $\Gamma$ if $a^{\star}_{i}$ is the optimal solution of the respective optimization problem of \eqref{eq:cost-function-game}, with $a_j$ fixed to $a_j^\star$ for all $ a_j \neq a_i $.
\end{definition}

\section{Analysis}

In this section, we establish the existence of a unique NE for the game $\Gamma$ as defined in Definition \ref{def:game}.

\subsection{Problem Overview}

In practice, a large number $N$ of heterogeneous mobile devices connect to the same MEC server to accelerate their inference tasks. These devices exhibit significant disparities in terms of computing capacities $E(d_i)$, network uplink rates $b(d_i)$, and inference tasks $\mathcal{G}_i$. As a result, a direct analysis of the strategy vector $\mathcal{A}$ becomes challenging due to \emph{the curse of dimensionality}.

To mitigate this, we introduce a term $A := \frac{\sum_{a_j \in \mathcal{A}} a_j}{S}$, representing the current unit price of MEC service when $A > 1$. The resource allocation $g_i$ for mobile device $d_i$ defined in Eq. \eqref{eq:res-proportion} can be reformulated as:
\begin{equation} \label{eq:res-proportion-mfg}
        {g}(a_i ;A) := \frac{a_i}{\max{(A, 1)}}
\end{equation}
where $\max(A, 1)$ signifies the current MEC service unit price. Each mobile device's optimization problem is thus to minimize its cost $\mathcal{L}(d_i, {a}_i; A)$. The introduction of $A$ brings two significant advantages: (1) \emph{The existence of NE $\mathcal{A}^\star$ is equivalent to the existence of the fixed point $A$, as $a^\star_i$ can be derived from Eq. \eqref{eq:res-proportion} and Eq. \eqref{eq:cost-function-game} once $A$ is known}, simplifying the mathematical analysis. (2) \emph{Compared to obtaining the strategy vector $\mathcal{A}$, acquiring the value of $A$ from the server is more straightforward for a mobile device}, hence making the practical implementation lightweight.

\subsection{The Property of the Cost Function}

We deploy a simple DNN with a structure identical to the one in Fig. \ref{fg:mincut-describe}, containing four neural layers, to a mobile device connected to a MEC server. Fig. \ref{fg:problem-curve} illustrates the cost function $\mathcal{L}(d_i, a_i; A)$ with an increasing $\frac{a_i}{S}$. The blue line represents the inference latency cost $T^\star(d_i, g_i)$, the orange line depicts the service charge $\beta (a_i)$, and the pink line indicates the overall cost $\mathcal{L}(d_i, {a}_i; A)$. The blue and pink lines can be divided into three segments, as shown in Fig. \ref{fg:problem-curve}.

Each segment symbolizes a different partition strategy $\mathcal{C}$. Strategy A (St. A) indicates that the entire DNN is executed locally, and $\mathcal{C} = \left\{ (v, s) | \forall v \in \mathcal{G} \right\} $. Since different $a_i$ values only affect the value of $t_c$ (the blue lines in Fig. \ref{fg:mincut-describe}), they do not impact the mincut value when devices perform local inference. As $a_i$ increases, the inference latency of layers on the server decreases. Segment St. B illustrates the cost when layers $v_1$ and $v_2$ are executed locally and the remaining layers are executed on the server. The corresponding partition strategy $\mathcal{C}$ is $\{ (v_1, v_1') \}$. As the allocated computing resource increases, the inference latency continues to decrease. The overall cost also declines, albeit at a slower rate, due to the increasing value of the charge function. Segment St. C represents all layers being assigned to the server, and $\mathcal{C} = \{(l, v) | \forall v \in \mathcal{G} \}$.

\begin{figure}
        \centering
        \subfloat[]{ \label{fg:problem-curve-parts}
                \includegraphics[width=0.8\linewidth]{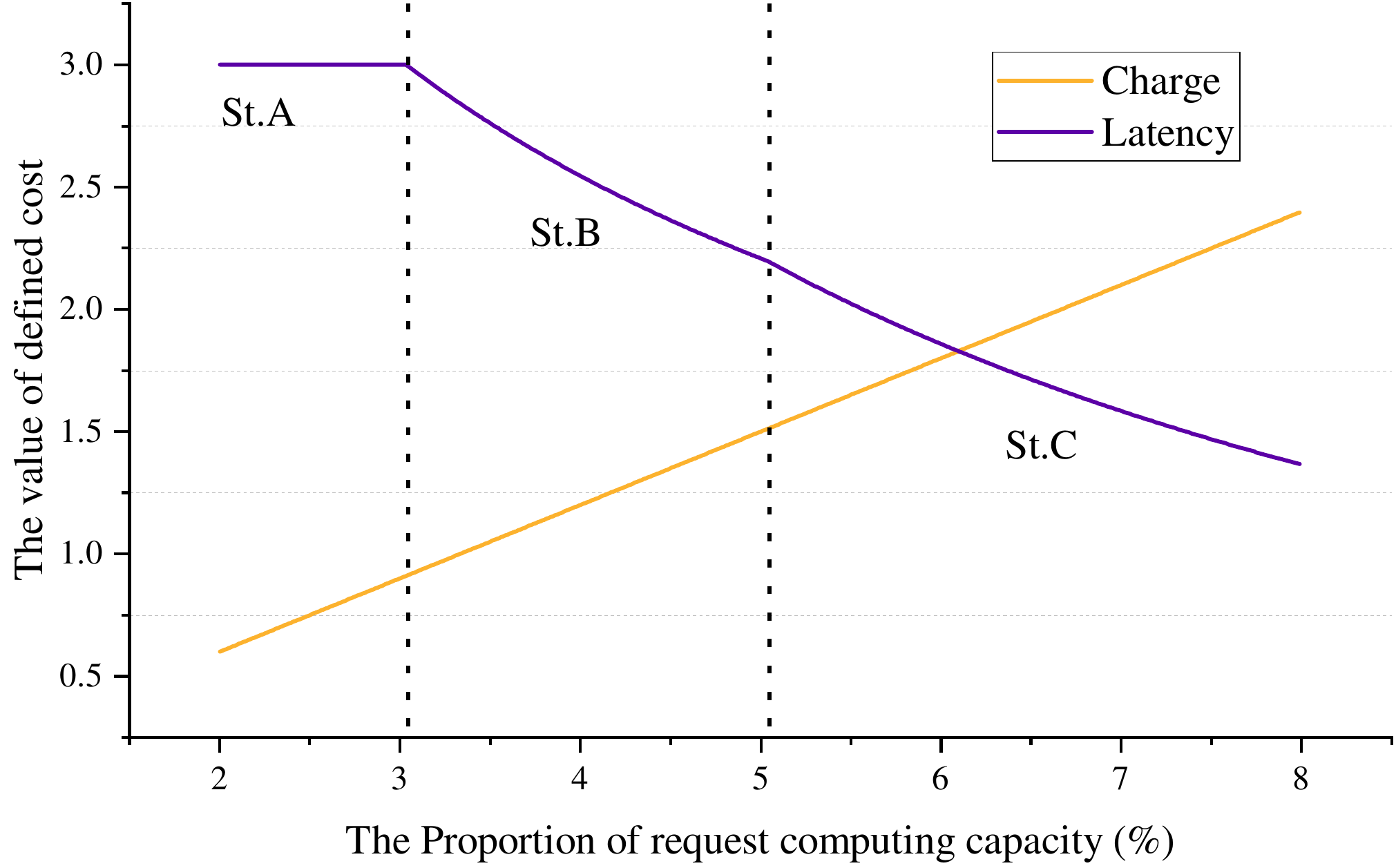}
        }

        \subfloat[]{ \label{fg:problem-curve-sum}
                \includegraphics[width=0.8\linewidth]{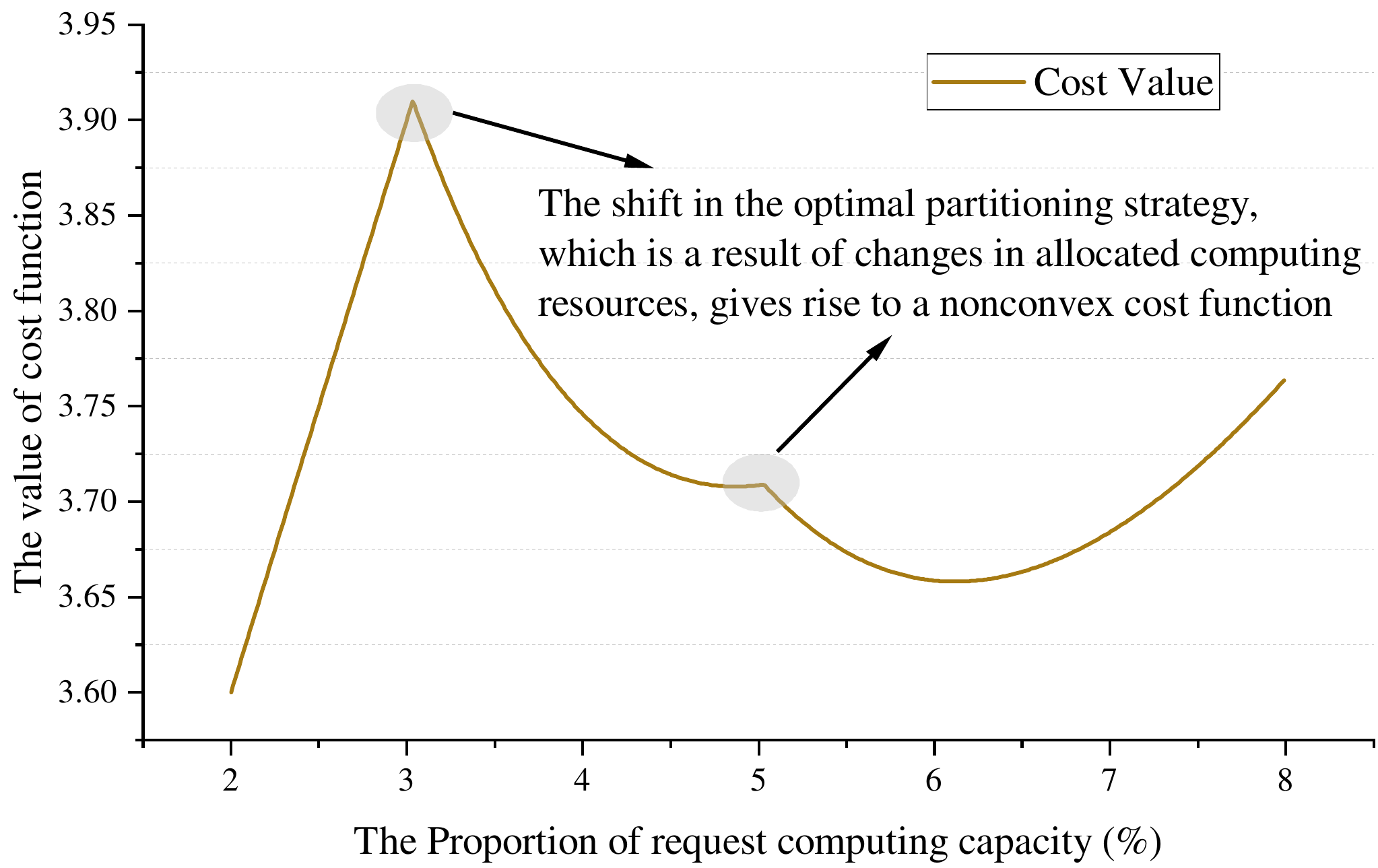}
        }
        \caption{The properties of the cost function $\mathcal{L}(d_i, a_i, \mathcal{A})$ with varying parameter $a_i$. The x-axis displays the value of $\frac{a_i}{S}$. (a) The changing of service charge $\gamma a_i$ and inference latency $T^\star(d_i, g_i)$ with different parameter $a_i$. (b) The changing of the cost function value $\mathcal{L}(d_i, a_i, \mathcal{A})$ with varying parameter $a_i$.}
        \label{fg:problem-curve}
\end{figure}

Here we focus on giving an accurate expression of cost function $\mathcal{L}$. From the above observation, $\mathcal{L}(d_i, a_i; A)$ is non-convex and non-smooth at these points that the optimal partition strategies switch. Since the cost function $\mathcal{L}(d_i, a_i; A)$ is composed by the inference latency function $T^\star(d_i, a_i)$ and the charge function $\beta(a_i)$. In Eq. \eqref{eq:cost-function-game}, $\beta(a_i)$ is set to be a linear function.  Therefore, we mainly discuss $T^\star(d_i, a_i)$ here, which returns the minimum edge cut value of a latency graph $\mathcal{\widehat{G}}$ defined in Section \ref{sec:graph-mincut}.

For any graph $\mathcal{\hat{G}} := <\mathcal{\hat{V}}, \mathcal{\hat{E}}, \phi>$, the minimum edge cut problem can be formulated with the following integer programming \cite{trevisan_2011}:
\begin{equation}
        \begin{array}{cl}
                \multicolumn{2}{c}{\min \sum_{e \in \mathcal{\hat{E}}} \phi(e) z_{e}}            \\
                \multicolumn{2}{c}{\text { s.t. }}
                \\
                w_{v}+ z_e = 1      & \forall e:(s, v) \in \mathcal{\hat{E}}                     \\
                w_v - w_u + z_e = 0 & \forall e:(u, v) \in \mathcal{\hat{E}}, u \neq l, v \neq s \\
                -w_u + z_e  = 0     & \forall e:(u, t) \in \mathcal{\hat{E}}
        \end{array}
\end{equation}
where $\mathbf{w}: \{w_v\} \in R^{|\mathcal{\hat{V}}| - 2}$ and $\mathbf{z}: \{ z_e \} \in R^{|\mathcal{\hat{E}}|}$ are two vectors that every element $z_{e}, w_v \in \left\{ 0, 1 \right\}$, and there is no $w_l$ or $w_s$ for source or sink vertex in $\mathbf{w}$. If the edge $e$ in the edge cut $\mathcal{\widehat{C}}$, we have $z_e = 1$. Hence, every edge cut $\mathcal{\widehat{C}}$ corresponds to a certain $\mathbf{z}$. Similarly, $w_v = 1$ indicates that vertex $v$ is in the same set with the sink $s$, otherwise, it is with the source $l$.

Let $\mathcal{\hat{E}}' := \{(l, v) | \forall v \in \mathcal{\hat{V}} \}$ denote all the edges that represent the inference latency if layer $v$ is deployed on the MEC server. The cost function $\mathcal{L}(d_i, a_i; A)$ can be expressed as:
\begin{equation} \label{eq:cost-function-lp}
        \begin{array}{cl}
                \multicolumn{2}{c}{\min \left\{
                        \sum_{e_i \in \mathcal{\hat{E}} - \mathcal{\hat{E}}'} \phi(e_i) z_{e_i}
                        + \sum_{e_j:(l, v) \in \mathcal{\hat{E}}'} \frac{f(v)}{g(a_i; A)} z_{e_j}
                \right\} + \gamma a_i }                                                          \\
                \multicolumn{2}{c}{\text { s.t. }}
                \\
                w_{v}+ z_e = 1      & \forall e:(s, v) \in \mathcal{\hat{E}}                     \\
                w_v - w_u + z_e = 0 & \forall e:(u, v) \in \mathcal{\hat{E}}, u \neq l, v \neq s \\
                -w_u + z_e  = 0     & \forall e:(u, t) \in \mathcal{\hat{E}} \text{.}
        \end{array}
\end{equation}
We will not attempt to directly solve the equation, as integer programming is NP-complete. Eq. \eqref{eq:cost-function-lp} is only used for mathematical analysis.

\begin{lemma} \label{th:A->a}
        The optimal $a_i^\star$, which minimizes the cost function $\mathcal{L}(d_i, \cdot; A)$ with a fixed $A$, is given by $a_i^\star = \sqrt{\frac{c_i^\star \max(A, 1)}{\gamma }}$, where $c_i^\star = \sum_{e:(l, v) \in \mathcal{\hat{E}}' } f(v) z_{e}^\star $ and $z_{e}^\star \in \{0, 1\}$ denotes whether $e$ is included in the minimum edge cut.
\end{lemma}

\begin{proof}
        Since $\mathcal{L}(d_i, a_i; A)$ is a piecewise function and is undifferentiable at the piecewise points, we first demonstrate that the optimal $a_i^\star$ that minimizes $\mathcal{L}$ cannot be at these differentiable points. If $a_k$ is a piecewise point that minimizes $\mathcal{L}$, then the left derivative $\mathcal{L}'({a}^{-}_k) < 0$ and right derivative $\mathcal{L}'({a}^{+}_k) > 0$. Note that the piecewise point also implies a switch in the optimal partition strategy $\mathcal{C}^\star$. However, the first order derivative of $\gamma \beta'(\cdot) = \gamma > 0$ holds for every $a_i$, hence $T^{\star'}(a^{-}_k) < \beta'(a_k) < 0$. If the partition strategy $\mathcal{C}$ does not switch at $a_k$, there must exist a small $\Delta a$ such that $T^\star(g(a_k + \Delta a; A)) > \widetilde{T}^\star(g(a_k + \Delta a; A))$, where $\widetilde{T}$ denotes the inference latency if the strategy $\mathcal{C}^\star$ does not switch. Since $T^\star$ should return the minimum inference latency, $a_i^\star$ cannot be the piecewise point. In addition, $a_i = 0$ implies that device $d_i$ does not participate in the game $\Gamma$ and makes no contribution to $A$, and $a_i = S$ implies there is only one player in the game. Both cases contradict our assumptions, thus the location of optimal $a_i^\star$ must be differentiable.

        For the optimal $a_i^\star$, the cost function $\mathcal{L}(d_i, {a}i^\star; A)$ can be represented as $\frac{c^\star_i}{g(a_i^\star; A)} + k + \gamma a_i^\star$, where $k = \sum_{e_i \in \mathcal{\hat{E}} - \mathcal{\hat{E}}'} \phi(e_i) z_{e_i}^\star$ is a constant. By setting $\mathcal{L}'(d_i, {a}_i^\star; A) = 0$, we derive $a_i^\star = \sqrt{\frac{c_i^\star \max(A, 1)}{\gamma }}$.
\end{proof}

\subsection{Existence of a Unique Nash Equilibrium}

The existence of a unique Nash Equilibrium (NE) in game $\Gamma$ can be established by constructing a contraction mapping on $A$ and applying the Banach fixed-point theorem \cite{ciesielski2007stefan}. We use a cumulative distribution function $F: [0, 1] \to [0, 1]$ to describe the distribution of $a_i$. If all $d_i$ are arranged in a specific order, such as sorted by descending local computing resources, $F$ can be defined as $F(s; A) := \frac{\sum_{s}a(d_i, A)}{AS}$, where the indices of these devices $d_i$ are less than $sN$.

The interpretation of $F(s; A)$ is that a $s$ percentage of mobile devices contributes a $F(s; A)$ percentage income to the MEC server when the current unit price is $A$. With the cumulative distribution function $F$, we can deduce that:
\begin{lemma} \label{th:contraction}
        If $\gamma > \frac{c^\star}{4S^2}$, then $F(s; \cdot)$ is a contraction mapping with respect to $A$, where $c^\star := \max{c^\star_i}$.
\end{lemma}

\begin{proof}
        From the characterization of $a_i$ revealed by Lemma \ref{th:A->a}, we can consider $a_i$ as a function of $a(d_i, A)$. We begin by differentiating $a(d_i, A)$ with respect to $A$, and the first-order condition shows that $ \frac{\mathrm{d} a(d_i, A)}{\mathrm{d} A} \leq \sqrt{\frac{c^\star_i}{4 A \gamma}} < \sqrt{\frac{c^\star_i}{4 \gamma}} < \sqrt{\frac{c^\star}{4 \gamma}}$, where $c^\star$ is a constant representing the maximum of all $c^\star_i$. If we set $\gamma > \frac{c^\star }{4 S^2}$, we obtain $ \frac{\mathrm{d} a(d_i, A)}{\mathrm{d} AS} < \sqrt{\frac{c^\star}{4\gamma S^2}} < 1$ for every $d_i$. Hence, for every $A, s, h > 0$ and $0 < k = \sqrt{\frac{c^\star}{4\gamma S^2}} < 1$, we have:
        \[
                \begin{gathered}
                        F(s ; A\!+\!h) \!-\! F(s; A)\! =\! \frac{\sum_{s}{a(d_{i}, A\!+\!h)}}{(A\!+\!h)S} - \! \frac{\sum_{s}{a(d_{i}, A)}}{AS}\leq kh\text{.}
                \end{gathered}
        \]
        Thus, $F(s, \cdot)$ is a contraction mapping.
\end{proof}

Now, using Lemma \ref{th:contraction} and the Banach fixed-point theorem \cite{ciesielski2007stefan}, the existence of a unique NE can be deduced:
\begin{theorem}[Equilibrium of $A$] \label{th:existence}
        There exists a unique Nash Equilibrium $A^\star$ for game $\Gamma$.
\end{theorem}
Theorem \ref{th:existence} indicates that mobile devices can agree on a unit service price $A^\star$. Given the known cost function \eqref{eq:cost-function-lp}, $a_i^\star = \arg\min(\mathcal{L}(d_i, \cdot; A^\star))$ for each mobile device. Therefore, we can obtain the corresponding strategy vector $\mathcal{A}^\star$ for game $\Gamma$.

\section{Decentralized DNN Surgery}

In this section, we present our proposed \textit{Decentralized DNN Surgery} framework which adaptively adjusts the demand of each mobile device and chooses the best partition strategy in a decentralized manner.

\subsection{Analysis of Convergence Procedures}

While Theorem \ref{th:existence} demonstrates the existence of a unique Nash Equilibrium (NE) for mobile devices, finding procedures that converge to NE remains challenging. Such procedures have been identified only for specific game classes: potential games, 2-player zero-sum games, among a few others \cite{hu2003nash,hart2013simple}. However, numerous problems, including 2-player zero-sum, are generally \emph{hopelessly impractical to solve} \cite{shoham2008multiagent} due to their PPAD-hard nature \cite{daskalakis2009complexity,mehta2018constant}. We analyze the resource allocation game $\Gamma$ proposed in this paper, demonstrating its PPAD-hard characteristics.

\begin{theorem} \label{th:ppad-hard}
        The game $\Gamma$ proposed in this paper is PPAD-hard.
\end{theorem}
\begin{proof}
        The result can be easily inferred by showing that a degenerate version of game $\Gamma$ resembles a 2-player general-sum game, which is established as PPAD-hard \cite{mehta2018constant}. A key property of a 2-player general-sum game is the uncertainty of the total payoff. In other words, a player's payoff does not necessarily bear a direct or inverse proportion relationship with the other player's payoff. Suppose there are only two mobile devices in $\Gamma$. Given the complexity of the DNN structure, the cost function of game $\Gamma$ is nonconvex. Thus, the cumulative cost function values of these mobile devices are unpredictable. For instance, if one device opts for a significant portion (99\%) of computing resources from the MEC server, both devices stand to lose due to the higher price of the MEC service. Conversely, if one device chooses not to use MEC service (0\%), the payoff of this device could either increase or decrease, depending on the specific DNN and device configuration. However, the other device will invariably benefit from the decrease in service price.
\end{proof}

\subsection{ Iterative Approximation Algorithm}

\begin{algorithm}[tb]
        \caption{Iterative approximation algorithm on device $i$}
        \small
        \label{al:dp-homo}

        \KwIn{An initial $a_i^0$ for the mobile device $d_i$}

        \SetKwFunction{FMain}{Optimization}
        \SetKwFunction{FInference}{Inference}
        \SetKwProg{Fn}{Function}{:}{}

        \Fn{\FMain{}}{
        Generate DAG $\mathcal{G}$ using given DNN\;
        Send $a_i^0$ to the server\;
        Query the server to get the value of $ A^0$\;
        \For{$t \gets [0 \cdots T]$}{
        \If{$a_i^t = 0$}{
        \If{perform locally inference for a while}{
        Query the server to get the value of $ A^t$\;
        Run a grad search for $a_i^{t+1}$\;
        }
        }
        \Else{
        Estimate $g_i^t$ with Eq. \eqref{eq:res-proportion-mfg}\;
        Build the latency graph $\mathcal{\hat{G}}$\;
        Generate partition strategy $\mathcal{C}$ and $c_i(a_i^t, A^t)$ \; %
        \If{$c_i(a_i^t, A^t) = 0$}{
        $a_i^{t+1} \gets 0$\;
        \FInference{$\mathcal{C},0$}
        }
        \Else{
        Update $a_i^{t+1}$ using Eq. \eqref{eq:perfect-goal-itera} and Eq. \eqref{eq:perfect-goal-iterb}\;
        $ A^{t+1} \gets$ \FInference{$\mathcal{C},a_i^{t+1}$}\;
        }
        }
        }
        }

        \Fn{\FInference{$\mathcal{C}, a_i^{t+1}$}}{
        \If{$a^t_i = 0$}{
                Executing inference task $\mathcal{G}$ locally\;
        }
        \Else{
        $\mathcal{G}_l \gets$ partition the model $\mathcal{G}$ using $\mathcal{C}$\;
        Executing inference task using $\mathcal{G}_l$\;
        Upload $\mathcal{C}$, output of $\mathcal{G}_l$ and $a_i^{t+1}$ to server\;
        Receive inference result and $A^{t+1}$ from server\;
        }
        }
\end{algorithm}

Theorem \ref{th:ppad-hard} highlights the challenge of identifying an algorithm capable of precisely determining the value of NE. Consequently, we adopt a gradient-based iterative algorithm to approximate NE during inference. Beginning with an initial $a_i^0$, the algorithm's core concept involves each mobile device $d_i$ updating its $a_i^t$ based on the term $A^t$ at iteration $t$. The modification of $a_i^t$ subsequently triggers changes in $A^{t+1}$ during the next iteration. This process iterates until both $a_i$ and $A$ converge.

The choice of a gradient-based algorithm is justified for several reasons. Firstly, gradient-based algorithms have showcased remarkable effectiveness across a wide array of empirical challenges, encompassing machine learning, deep learning, and game theory \cite{balduzzi2018mechanics}, all of which frequently present complex optimization problems. Secondly, each part in $\mathcal{L}$ can be interpreted as the sum of a reciprocal function and a linear function, as depicted in Eq. \eqref{eq:cost-function-lp} and Fig. \ref{fg:problem-curve}. Thus, the cost function $\mathcal{L}$ is differentiable, except at piecewise points. The proof procedure of Lemma \ref{th:A->a} demonstrates that these piecewise points cannot represent the optimal $a^\star$. Consequently, we can assign a subgradient value to these piecewise points, an approach extensively adopted within the deep learning community (e.g., ReLU \cite{agarap2018deep}). Thirdly, the computation cost for a gradient-based algorithm is minimal due to the design where a mobile device receives a scalar $A^t$ from the server and does not directly interact with others. All parameters in Eq. \eqref{eq:cost-function-lp} are scalars, not tensors, thereby rendering the gradient computation lightweight.

The gradient $\nabla a_i$ can be calculated from Eq. \eqref{eq:cost-function-lp} as follows:
\begin{equation} \label{eq:perfect-goal-itera}
        \nabla a_i =  \frac{\mathrm{d} \mathcal{L}}{\mathrm{d} a_i } = \gamma - \frac{c_i(a_i, A)}{\max(1, A) {g_i}^2 }.
\end{equation}
Here, $c_i(a_i, A):= \sum_{(l, v) \in (\mathcal{\hat{E}}' \cup \mathcal{\hat{C}}) } f(v)$ is analogous to $a_i^\star$ as defined in Lemma \ref{th:A->a}. $\mathcal{\hat{C}}$ represents the min-cut of $\mathcal{\hat{G}}$ when $a_i$ and $A$ are given. $c_i(a_i, A)$ physically signifies the sum of FLOPs required by the layers deployed on the MEC server, calculated using an efficient minimum edge cut algorithm. We discuss computational complexity later in Section \ref{sec:complexity}. The execution time of DNN inference is relatively short. For instance, YOLOv2 inference takes roughly 2s and 0.24s on the mobile device and the cloud server, respectively \cite{infocom19:mincut}. The gradient-based algorithm is efficient enough and minimally impacts inference performance.

Since the cost function is non-convex, $a_i$ is updated with momentum in order to jump over the local minimum point. The update formula is written as follows:
\begin{equation} \label{eq:perfect-goal-iterb}
        \begin{array}{cl}
                \nu^t   & = \rho \nu^{t-1} + (1-\rho) \nabla a^{t}_i \\
                a_i^{t} & = a_i^{t-1} - \eta \nu^t \text{,}
        \end{array}
\end{equation}
where $\nu^t$ is the momentum at iteration $t$, $\rho$ is an exponential decay factor and $\eta$ is the learning rate.

Only using the gradient-based algorithm to update $a_i$ can lead to an incorrect result. When the current $A$ is large, $a_i$ will decrease as $\nabla a > 0$. After updating $a_i$, the corresponding partition strategy may execute the entire DNN locally, namely $c_i(a_i, A) = 0$. Thus, $\nabla a$ will always equal to $\gamma > 0$, and $a_i$ will reach $0$ in the end even though $A$ has decreased in the following iterations.

\textit{Resource Sniff} is proposed to deal with this problem. When $c_i(a_i, A) = 0$, the mobile device $d_i$ do not use the computing resource and locally execute the inference task, namely $a_i=0$. The mobile device $d_i$ will periodically run a grid search for $a_i$ after every several iterations until a better $a_i>0$ is found. Since there is only one parameter $a_i$ during the grid search, the sniff operation is also quite lightweight. The full version of the iterative approximation algorithm is shown in Algorithm \ref{al:dp-homo}. Note the term $A^{t+1}$ is received along with the current inference result (Line 19), thus DDS introduces nearly zero additional communication considering the size of the floating number $A$.

\subsection{ Computational Complexity} \label{sec:complexity}
In this section, we examine the computational complexity of the DDS framework. The DDS workflow incorporates two additional computational costs beyond the inference task. Firstly, the computation associated with estimating the computing resources $\hat{a}_i^t$ (as per Line 11 in Algorithm \ref{al:dp-homo}) and updating the requested computing resource $a_i^{t+1}$ for the subsequent inference (Line 18 in Algorithm \ref{al:dp-homo}). Secondly, the computation of determining the minimal edge cut $\mathcal{C}$, essentially the optimal partition strategy, for DNN graph $\mathcal{G}$.

The first computational cost involves Eq. \eqref{eq:res-proportion-mfg}, \eqref{eq:perfect-goal-itera}, and \eqref{eq:perfect-goal-iterb}. All parameters including $a_i$, $A$, $\gamma$, and $c_i(a_i, A)$ are scalars, not vectors. This characteristic stems from DDS's decentralized design: MEC server aggregates the actions of all other devices (corresponding to $A$), and each mobile device communicates solely with the MEC server. Thus, in comparison to other $N$-Player games \cite{balduzzi2018mechanics}, the decision-making process of mobile devices in DDS only interacts with $A$, provided by the MEC server. This feature suggests that the computational complexity of the first cost is $O(E)$, remaining constant as the number of mobile devices increases. Furthermore, the calculations in Eq. \eqref{eq:res-proportion-mfg}, \eqref{eq:perfect-goal-itera}, and \eqref{eq:perfect-goal-iterb} involve only a few concise algebraic operations, in contrast to a small fully-connected layer in a DNN, which could necessitate millions of such operations. Therefore, the computational cost of these equations can be considered negligible.

The second computational cost involves calculating the minimum cut of the latency graph $\mathcal{\hat{G}}$ from the given DNN $\mathcal{G}$. As this computation does not include other mobile devices, the complexity relates solely to the properties of $\mathcal{G}$. The classic Boykov's algorithm \cite{boykov2004experimental} presents a complexity of $O\left((V + E)V^2\right)$, where $V$ and $E$ denote the number of vertices and edges in the latency graph $\mathcal{\hat{G}}$ respectively. Due to the property of the DNN graph containing multiple cutting vertices, which divide the DNN into several subgraphs (e.g., ResBlock, InceptionBlock), the computational complexity has been reduced to $O\left((k+1)(\hat{V}+\hat{E}) \hat{V}^2\right)$, where $k$ represents the number of cutting vertices in DNN $\mathcal{G}$. $\hat{V}$ and $\hat{E}$ signify the number of vertices and edges in these subgraphs. For a more detailed exploration, we direct the reader to \cite{ubicom20:faster-infer}. The decision time of the optimized minimal edge cut algorithm is approximately 0.1 seconds even on resource-limited devices like RaspberryPi \cite{ubicom20:faster-infer}.

In summary, the computational complexity of DDS framework is $O\left(E + (k+1)(\hat{V}+\hat{E}) \hat{V}^2\right)$ for DNN $\mathcal{G}$, and $O(1)$ for the number of mobile devices $N$.

\section{ Performance Evaluation}

In this section, we evaluate the performance of DDS.

\subsection{ Environment Setup}

\begin{table}[tb]
        \caption{Metrics Of Different DNNs } \label{tb:dnn-metrics}
        \centering
        \begin{tabular}{@{}l|c|c|c|c@{}}
                \toprule
                Model    & VGG11 & ResNet34 & ResNet50 & \ \ ViT \ \ \\ \midrule
                Vertices & 15    & 55       & 73       & 26          \\
                Edges    & 14    & 57       & 75       & 32          \\
                GFLOPs   & 7.63  & 3.68     & 4.12     & 3.47        \\ \bottomrule
        \end{tabular}
\end{table}

\subsubsection{ Platform}
The DDS framework is implemented using the PyTorch deep learning framework \cite{pytorch}. gRPC \cite{grpc} is utilized for efficient communication between mobile devices and the MEC server. The source code of PyTorch is modified based on \cite{narayanan2019pipedream} to automatically convert DNNs into a DAG structure. Evaluations of the DDS convergence performance with a large scale (up to 100) of mobile devices are simulated on a machine equipped with an Intel Xeon E5-2660 CPU and 256 GB memory. The configurations of these mobile devices, such as bandwidth, computing resource, and inference tasks, are randomly chosen from $5 \sim 10$ Mbps, $10 \sim 20$ GFLOPS, and several DNN models. The computing power of the MEC server is set to 1.2 TFLOPS.

\subsubsection{Tasks}
We employ four real-world DNNs: VGG11 \cite{vgg}, ResNet34 \cite{resnet}, ResNet50 \cite{resnet}, and the Visual Transformer (ViT) \cite{vit}. All input data for these models are sampled from the ImageNet dataset \cite{imagenet}. The metrics of these DNNs, including the number of vertices (neural layers), edges, and required GFLOPs, are listed in Table \ref{tb:dnn-metrics}.

\subsubsection{Baseline}
Three baselines are used in the experiment:
\begin{itemize}
        \item Edge Only (EO): The entire DNN is executed locally on the mobile device.
        \item Server Only (SO): All captured raw data are directly uploaded and the DNN is executed on the MEC server.
        \item Dynamic Adaptive DNN Surgery (DADS): A min-cut algorithm is used to find the best partition strategy \cite{infocom19:mincut} for fixed computing resources.
\end{itemize}
The computing resource allocated to each mobile device for the three baselines is set to $\frac{S}{N}$ during inference.

\begin{figure}[tb]
        \centering
        \subfloat[]{ \label{fg:convergence}
                \includegraphics[width=0.8\linewidth]{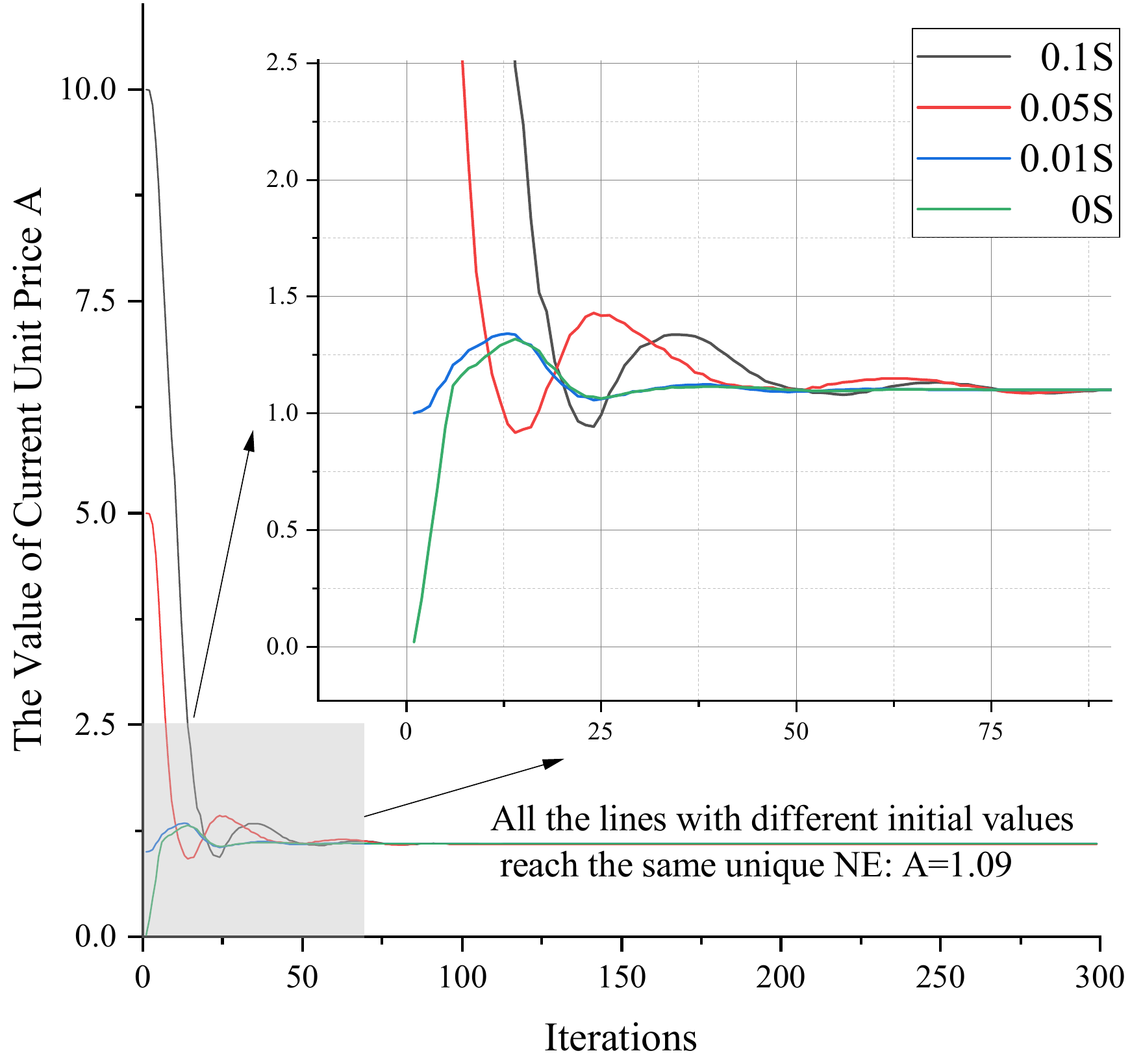}
        }

        \subfloat[] { \label{fg:node-detail}
                \includegraphics[width=0.8\linewidth]{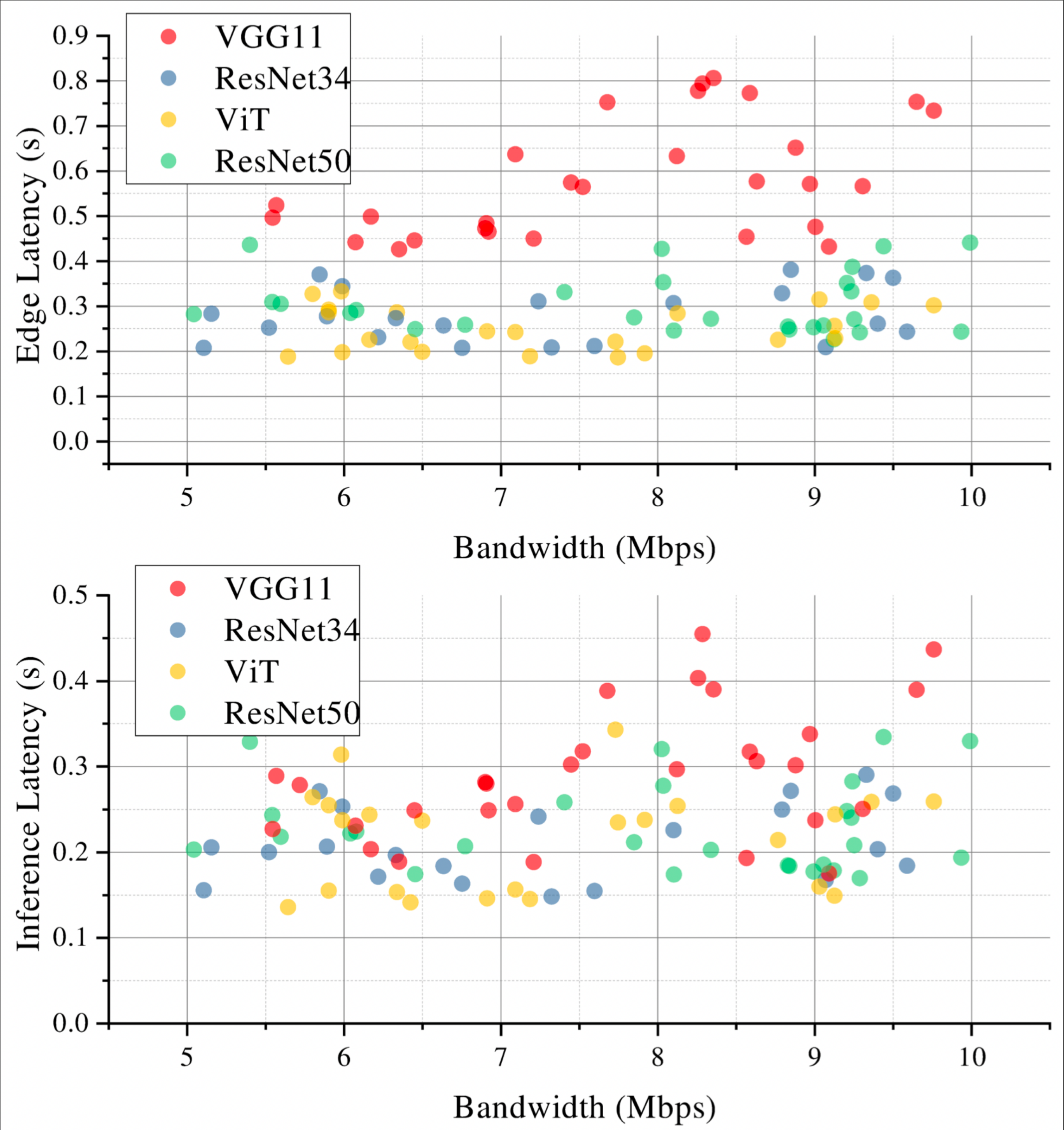}
        }
        \caption{(a) The value of $A^t$ with different initial values $a_0$ at different iterations from the perspective of a specific device. (b) The local inference latency and inference latency after using the service of MEC server for every mobile device.}

\end{figure}

\subsection{Convergence Performance}

To examine the convergence performance of our proposed iterative algorithm, we simulate a scenario with 100 mobile devices, all sharing the same initial budget $a^0$. Given the asynchronous nature of our algorithm across the devices, we randomly select one device, $d_i$, and track the value of $A^t$ as perceived by $d_i$ at varying iterations. We carry out simulations under four distinct settings of the initial $a^0$ and plot the results in Fig. \ref{fg:convergence}.

While gradient-based algorithms do not guarantee convergence to the unique NE, our simulation demonstrates that mobile devices with differing initial $a^0$ values rapidly reach consensus on $A^t$ within a few iterations. In our experiment, this consensus is the service price of $1.09$, which all mobile devices accept. Notably, the convergence process is expedient; it requires fewer than 25 iterations for convergence among 100 mobile devices when $a^0$ is set to $0$ or $0.01S$. Even when $a^0$ is assigned impractical values, such as $0.05S$ and $0.1$, convergence is still achieved within 50 iterations. These results underscore the efficacy of the proposed iterative approximation algorithm.

When $a^0$ is initially set to $0.1S$ or $0.05S$, $A^t$ declines rapidly at the onset of iterations due to the high charge, $\beta(a_i^t)$. If $a^0 = 0$, all mobile devices incrementally explore the MEC server's computing resources and select an appropriate $a_i^t$ for iteration, causing $A^t$ to gradually increase initially. It is worth noting that $A^t$ exhibits a fluctuating pattern before finally reaching the NE during iterations, as the momentum strategy is employed for updating $a^t_i$. This momentum strategy expedites $A^t$ optimization when the optimization direction is unambiguous (during the early stages of $0.1S$ and $0.05S$) and enhances convergence performance by circumventing some local minima.

\subsection{Resource Allocation Across Various Devices}

Fig. \ref{fg:node-detail} provides a detailed depiction of resource allocation within the DDS, encompassing 100 heterogeneous mobile devices. Each point within the figure represents a mobile device. Given the diversity in uplink bandwidths $b(d_i)$, local computing capacities $E(d_i)$, and inference tasks $\mathcal{G}_i$ among the devices, each point encapsulates multiple dimensions of information. The color of each point indicates the DNN $\mathcal{G}_i$ assigned to device $d_i$. The x-axis represents the bandwidth $b(d_i)$ of different devices, while the y-axis of the upper figure illustrates the inference latency if $\mathcal{G}_i$ is executed locally. The lower figure displays the inference latency post-utilization of the MEC service. The variation in inference latency between the two figures suggests the allocation of resources from the MEC server, responding to the diverse demands of the mobile devices.

We can find that the execution time of mobile devices differs for the four DNNs in the top figure. The inference of VGG11 is most time-consuming and takes between 0.5s and 0.8s. The other DNNs (ResNet34, ResNet50, and ViT) take around 0.3s as their required GFLOPs of them are similar. Meanwhile, the inference latency of VGG11 reduces up to 0.5s after adopting DDS. The mobile devices with VGG11 tasks are willing to reduce the high inference latency $T^\star$ with some additional service charge. The inference latencies of ResNet34, ResNet50, and ViT also decrease, but the reduction is not as obvious as that of VGG16. Since the GFLOPs of these DNNs are smaller than VGG11, the gain of reducing inference latency can not remedy the increase of charge. As a result, they use smaller computing resources.

\subsection{Average Inference Latency}

\begin{figure*}[tb]
        \centering
        \subfloat[Inference Latency]{ \label{fg:latency-detail-all}
                \includegraphics[width=0.23\linewidth]{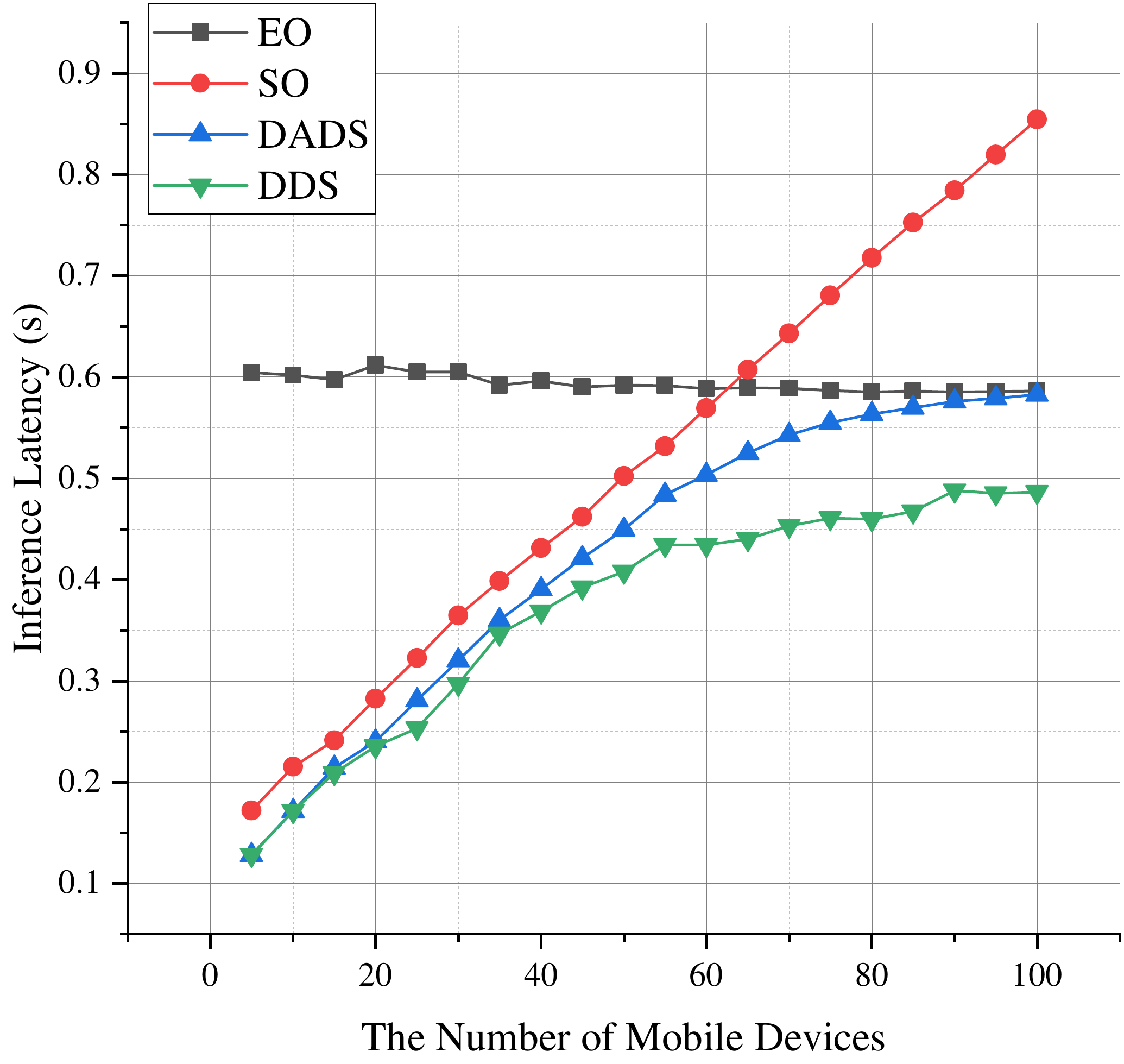}
        }
        \subfloat[Server Latency]{ \label{fg:latency-detail-cloud}
                \includegraphics[width=0.23\linewidth]{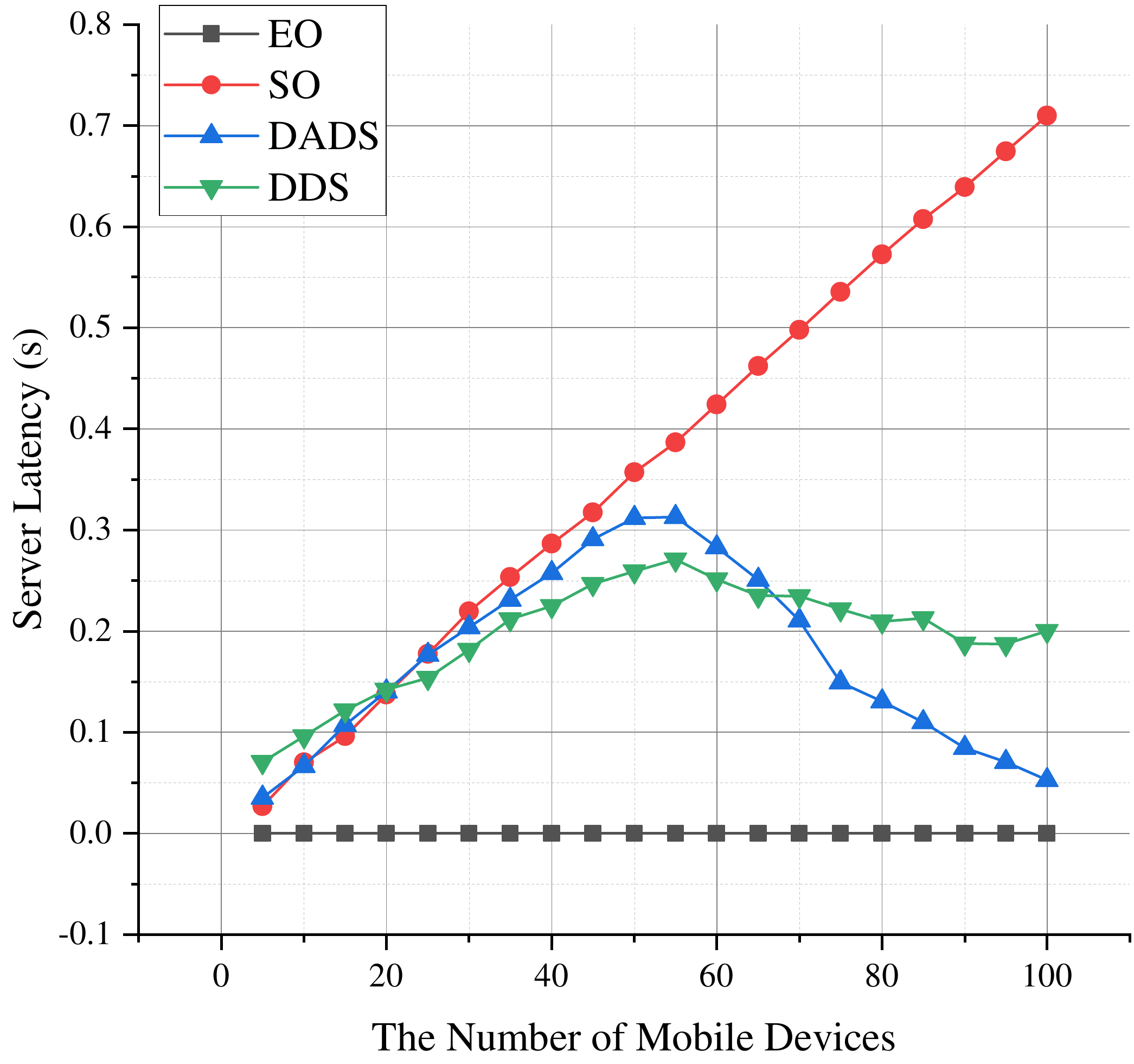}
        }
        \subfloat[Network Latency]{ \label{fg:latency-detail-comm}
                \includegraphics[width=0.23\linewidth]{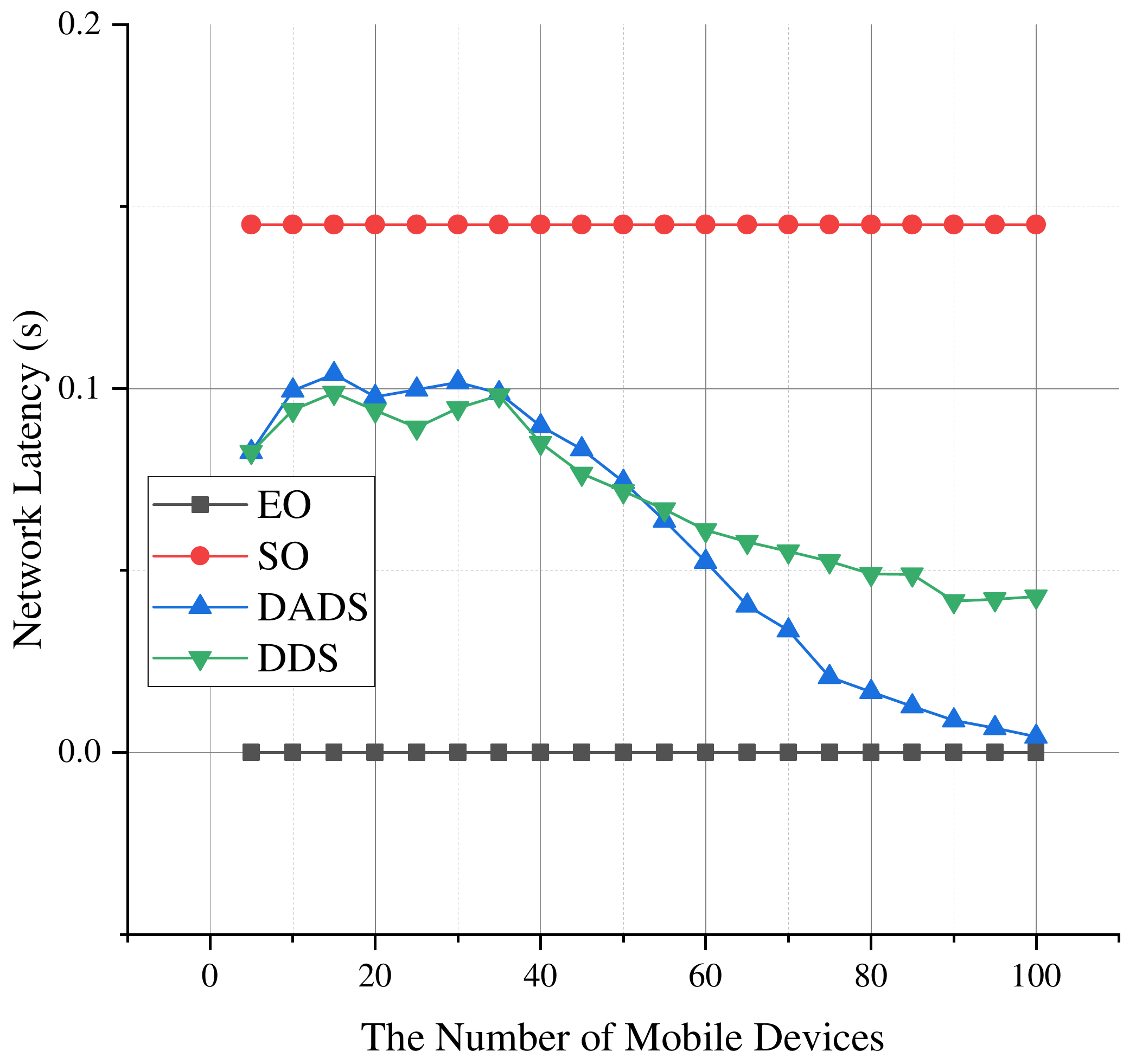}
        }
        \subfloat[Device Latency]{ \label{fg:latency-detail-edge}
                \includegraphics[width=0.23\linewidth]{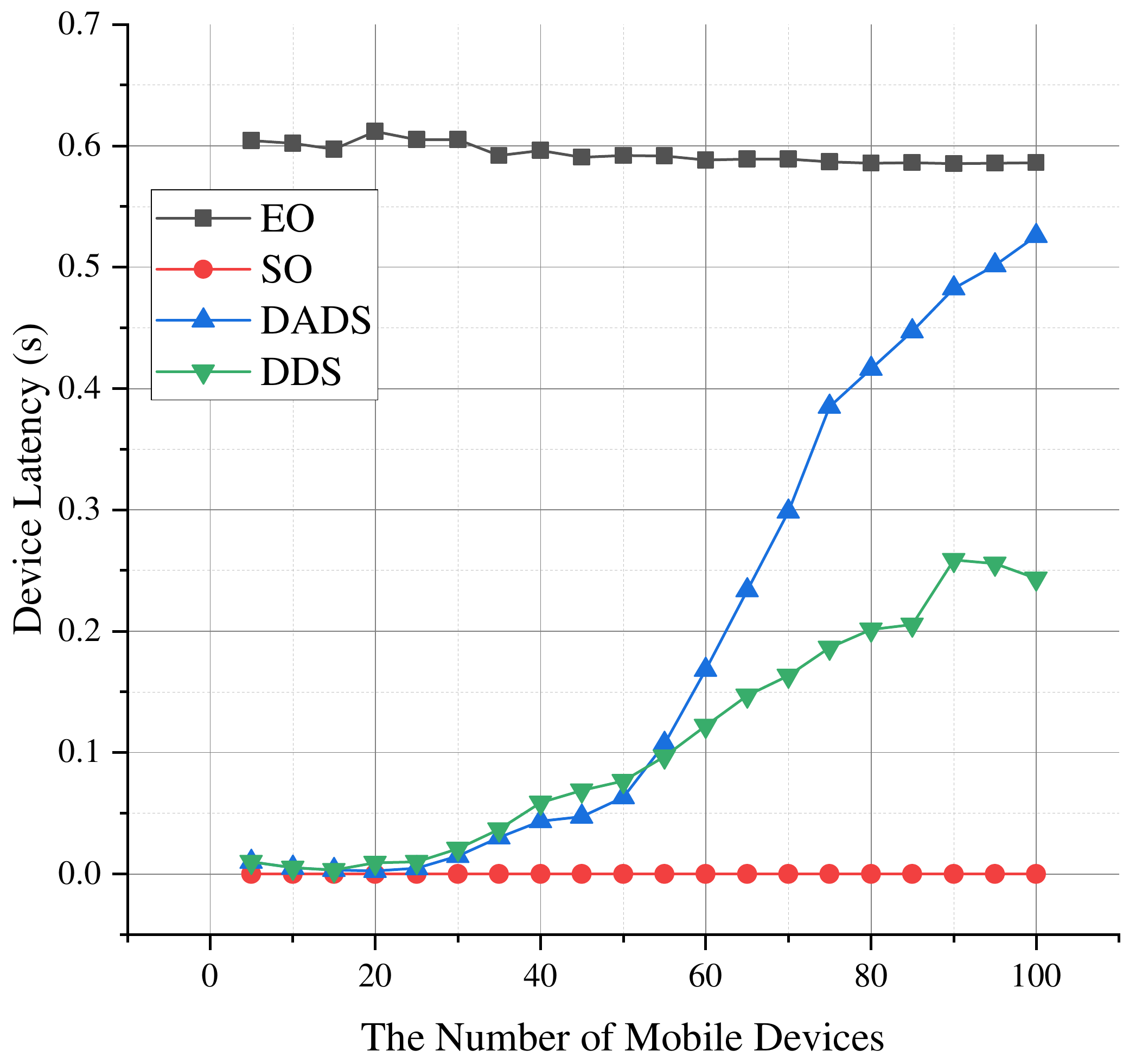}
        }
        \caption{Average latency trends as the number of mobile devices increases.
                (a): Overall inference latency $T$.
                (b): Server inference latency $T^s$.
                (c): Network latency $T^t$ for feature uploading.
                (d): Local device inference latency $T^l$ on the mobile device.
        }
        \label{fg:latency-detail}
\end{figure*}

We evaluate the average inference performance involving multiple mobile devices ranging from 5 to 100. Since the MEC server capacity is limited, the server workloads grows as the number of mobile devices increases. We capture the time cost at each inference stage. Figure \ref{fg:latency-detail} illustrates the trend of the average latency for all mobile devices using four distinct algorithms as the number of devices increases. Figure \ref{fg:latency-detail-all} displays the overall inference latency $T^\star$, which encompasses all stages. The remaining figures depict the server inference latency $T^s$, network latency $T^t$, and local device inference latency $T^l$, respectively.

The inference latency grows when the number of mobile devices increases except for EO. Since SO uploads all data to MEC server, the performance will quickly decrease when the workload of server grows. The latency of SO grows with a linear rate in Fig. \ref{fg:latency-detail-all} and is longer than EO when the number of devices exceeds $65$. DADS and DDS can adaptively adjust the partition strategy $\mathcal{C}$, thus the growth of inference latency is slower than SO. From Fig. \ref{fg:latency-detail-edge} and Fig. \ref{fg:latency-detail-cloud} we can find when the workload of MEC server becomes heavy (with $50$ devices), both algorithms gradually reduce the computation on the server and place more computation on the mobile devices. However, when the number of devices reaches $100$, DADS degenerates to EO. From Fig. \ref{fg:latency-detail-comm} we can find the average network latency is close to zero, which means it hardly uses the network and executes almost entire DNN locally. But DDS still benefits from MEC server, since the demands of every mobile device can be adaptively adjusted. Some devices with lightweight DNN and better local computing resources tend to reduce their demands and save server resources for more needed ones. Thus, it achieves a better latency performance than DADS (about $1.25 \times$).

\section{ Related Work}

\subsection{ DNN Surgery}
DNN surgery offers an efficient means to reduce inference latency without sacrificing accuracy. The concept was first introduced by Neurosurgeon \cite{sigarch:neuro}, which delegates parts of the DNN model to a more robust cloud platform. DADS \cite{infocom19:mincut} concentrates on finding optimal partition strategies for DNNs with Directed Acyclic Graph (DAG) structures, proposing a min-cut based algorithm. QDMP \cite{ubicom20:faster-infer} notes that applying the min-cut algorithm to the entire DNN can be time-consuming and puts forth an improved algorithm to reduce time complexity when finding the optimal strategy. However, these efforts only consider a single mobile device and assume that the computing resources allocated from the server remain fixed during optimization. Some researchers have employed multiple cooperative mobile devices to accelerate inference \cite{coedge,icdcs:pico,edgeflow}. For example, CoEdge \cite{coedge} deploys the full model on each device and divides the intermediate features. PICO \cite{icdcs:pico} organizes mobile devices into an inference pipeline by partitioning both neural layers and features. Similarly, EdgeFlow \cite{edgeflow} aligns with \cite{coedge} but can handle models with a graph structure. Despite these advancements, there is a dearth of research on accelerating inference for multiple uncooperative mobile devices.

\subsection{Task Offloading}
The offloading problem for general tasks involving multiple mobile devices has been studied for years \cite{tpds14:offloading, mc:offloading,infocom21:threshold,jsac:task-offloading}. The task models in \cite{tpds14:offloading,infocom21:threshold,jsac:task-offloading} are atomic and cannot be partitioned. \cite{jsac:task-offloading} considers the offloading problem in a Software Defined Networking (SDN) \cite{sdn} environment, with the offloading strategy determined by a central controller. However, since the inference typically takes less than a second, the centralized algorithm would double the network latency as the strategy must first be dispatched to mobile devices. \cite{mc:offloading} considers ideal offloading tasks that can be split in any proportion and provides an algorithm to find the Nash Equilibrium directly. However, many real-world tasks, such as inference tasks, do not satisfy this assumption. \cite{infocom21:threshold} employs an iterative algorithm to find the optimal threshold that, when the number of waiting tasks exceeds the threshold, requires devices to upload newly arrived tasks to the server. Moreover, the problems addressed in \cite{tpds14:offloading,mc:offloading,infocom21:threshold} are convex and easier to tackle.

\section{ Conclusion}

Current algorithms fall short in the MEC environment where numerous self-regulating mobile devices share the computing resources offered by an MEC server. To tackle this issue, our study reformulates the optimal partition strategy as a min-cut of a latency graph derived from the original DNN. Concurrently, we design a resource allocation problem to determine the optimal demand for each diverse mobile device and deduce the existence of the unique NE within the proposed game.

We introduce the Decentralized DNN Surgery (DDS) framework to efficiently and adaptively obtain the NE for every device in the MEC environment. The lightweight nature of DDS allows it to be deployed on each mobile device without introducing any additional communication overhead. Experimental results demonstrate the efficiency of DDS under various settings, showing a notable increase in average inference speed by 1.25 $\times$ compared to the state-of-the-art algorithm. Consequently, our study paves the way for a scalable and efficient solution for DNN inference in dynamic MEC scenarios.

\bibliographystyle{IEEEtran}
\bibliography{IEEEabrv,reference}

\end{document}